\documentclass[conference]{IEEEtran}
\IEEEoverridecommandlockouts
\usepackage{times}
\usepackage{helvet}
\usepackage{courier}
\usepackage{natbib}
\usepackage{colortbl,xcolor}
\usepackage{amsmath}
\usepackage{amssymb}
\usepackage{algorithm}
\usepackage{algorithmicx}
\usepackage[noend]{algpseudocode}

\newtheorem{theorem}{Theorem}
\newtheorem{lemma}[theorem]{Lemma}

\newtheorem{problem}{Problem}
\newtheorem{definition}{Definition}
\newtheorem{remark}{Remark}
\newtheorem{example}{Example}
\newtheorem{axiom}{Axiom}

\newcommand{\NewPrefer}{weakly prefer }
\newcommand{\NewPrefers}{weakly prefers }

\newcommand{\WC}{\mathit{WC}}
\newcommand{\Mon}{\mathit{Mon}}
\newcommand{\A}{\mathit{A}}
\newcommand{\Emb}{\mathit{Emb}}
\newcommand{\SA}{\mathit{SA}}
\newcommand{\GS}{\mathit{GS}}
\newcommand{\SGS}{\mathit{SGS}}

\newcommand{\SAp}{\mathit{SA'}}

\begin{document}
\title{Communities in Preference Networks:\\ Refined Axioms and Beyond}%\thanks{This work is sup-ported by everybody.}
\author{\IEEEauthorblockN{Gang Zeng\IEEEauthorrefmark{1}\IEEEauthorrefmark{2}\IEEEauthorrefmark{4},
Yuyi Wang\IEEEauthorrefmark{3}, Juhua Pu\IEEEauthorrefmark{4}, Xingwu Liu\IEEEauthorrefmark{1}\IEEEauthorrefmark{2}\thanks{*Correspondence should be addressed to Xingwu Liu.},
Xiaoming Sun\IEEEauthorrefmark{1}\IEEEauthorrefmark{2} and Jialin Zhang\IEEEauthorrefmark{1}\IEEEauthorrefmark{2}}
\IEEEauthorblockA{\IEEEauthorrefmark{1}University of Chinese Academy of Sciences, Beijing, 100049, China
}
%\IEEEauthorblockA{\IEEEauthorrefmark{2}Institute of Computing Technology, Chinese Academy of Sciences, 100190, Beijing, China}
\IEEEauthorblockA{\IEEEauthorrefmark{2}Institute of Computing Technology, Chinese Academy of Sciences, 100190, Beijing, China\\
{zenggang,liuxingwu,sunxiaoming,zhangjialin}@ict.ac.cn%\thanks{This work was supported in part by the National Natural Science Foundation of China Grant 61222202, 61433014, 61502449, 61602440 and the China National Program for support of Top-notch Young Professionals.}
}
\IEEEauthorblockA{\IEEEauthorrefmark{3}Distributed Computing Group,
ETH Z\"{u}rich\\
yuyiwang920@gmail.com%\thanks{the National Natural Science Foundation of China (Grants no. 11601375).}
}
\IEEEauthorblockA{\IEEEauthorrefmark{4}State Key Laboratory of Software Development Environment,
Beihang University\\
pujh@buaa.edu.cn\thanks{This work is partially supported by the National Key Reaearch and Development Program of China (2016YFB1000201),
National Natural Science Foundation of China (11601375,61420106013),
Science Foundation of Shen-zhen City in China (JCYJ20160419152942010),
and State Key Laboratory of Software Development Environment Open Fund (SKLSDE-2015ZX-25).}}
}

%\author{\IEEEauthorblockN{Gang Zeng}
%\IEEEauthorblockA{Institute of Computing Technology\\
%University of Chinese Academy of Sciences\\
%zenggang@ict.ac.cn
%}
%\and
%\IEEEauthorblockN{Yuyi Wang}
%\IEEEauthorblockA{Distributed Computing Group\\
%ETH Z\"{u}rich\\
%yuyiwang920@gmail.com
%
%}
%\and
%\IEEEauthorblockN{Juhua Pu}
%\IEEEauthorblockA{State Key Laboratory of Software Development Environment\\
%Beihang University\\
%pujh@buaa.edu.cn
%\thanks{The National Key Reaearch and Development Program of China (2016YFB1000200),
%National Natural Science Foundation of China (11601375,61420106013),
%Science Foundation of Shen-zhen City in China (JCYJ20160419152942010),
%State Key Laboratory of Software Development Environment Open Fund (SKLSDE-2015ZX-25)}
%}
%\and
%\IEEEauthorblockN{Xingwu Liu}
%\IEEEauthorblockA{Institute of Computing Technology\\
%Chinese Academy of Sciences\\
%liuxingwu@ict.ac.cn
%\thanks{the National Natural Science Foundation of China (Grants no. 11601375).}
%}
%\and
%\IEEEauthorblockN{Xiaoming Sun}
%\IEEEauthorblockA{Institute of Computing Technology\\
%Chinese Academy of Sciences\\
%sunxiaoming@ict.ac.cn
%}
%\and
%\IEEEauthorblockN{Jialin Zhang}
%\IEEEauthorblockA{Institute of Computing Technology\\
%Chinese Academy of Sciences\\
%zhangjialin@ict.ac.cn
%}
%}
\maketitle

\begin{abstract}
\cite{teng2016itcs} investigated essential requirements for communities in preference networks.
They defined six axioms on community functions, i.e., community detection rules.
Though having elegant properties, the practicality of this axiom system is compromised by the intractability of checking two critical axioms,
so no nontrivial consistent community function was reported in \citep{teng2016itcs}.
By adapting the two axioms in a natural way,
we propose two new axioms that are efficiently-checkable.
We show that most of the desirable properties of the original axiom system are preserved.
More importantly, the new axioms provide a general approach to constructing consistent community functions.
We further find a natural consistent community function that is also enumerable and samplable, answering an open problem in the literature.
\end{abstract}

\section{Introduction}
\noindent
Clustering individuals in a social network, called community detection,
is a fundamental task in graph mining and has been adequately studied.
Community detection has different forms,
depending on whether overlapping communities are allowed \citep{palla2005uncovering,baumes2005efficient,zhang2007identification,ahn2009communities},
whether hierarchical structures are taken into account \citep{Sibson1973SLINK},
and in which form of the data is provided, etc.
People proposed a number of algorithms to find communities,
based on different principles such as spectral clustering \citep{Hoffman1973Lower},
density-based methods \citep{Ester1996A},
modularity-based method \citep{newman2006modularity}.
%粗略说说community是什么
No matter which form and which algorithm we choose, a community is usually considered as a group of closely related individuals.
However, ``a group of closely related individuals'' is a rather rough concept,
and there does not exist a widely accepted definition of communities.
% The lack of ... results in ...

%graph partition \citep{Kernighan1973Heuristic},
%hierarchical clustering \citep{Sibson1973SLINK}, partition clustering \citep{macqueen1967},
%spectral clustering \citep{Hoffman1973Lower} and modularity \citep{newman2006modularity}
%or overlapping clustering (community), e.g. clique percolation method \citep{palla2005uncovering}, local optimizating on a given function \citep{baumes2005efficient},
%fuzzy c-means clustering \citep{zhang2007identification} and clustering on edges \citep{ahn2009communities}.

In this paper, we try to axiomatize the concept of communities.
We allow overlapping communities %, consider hierarchical structures
and assume the data is given in a preference network.
% ... 还可以多解释一下

\cite{kleinberg2003impossibility} developed an axiomatic
framework for non-overlapping clusterings and defined clustering functions whose inputs are
a set $V$ of $n$ individuals and pairwise distances between these $n$ individuals and whose output is a partition of these vertices (i.e., overlapping communities are not allowed).
%Based on Arrow's impossibility theorem on social choice \citep{Arrow1963Social},
He found that three desired clustering axioms (scale-invariance, richness and consistency) cannot be satisfied %by any clustering function
at the same time.

But this negative result did not prevent researchers from investigating axiomatic aspects of clusterings.
\cite{Marina2005Comparing} proposed axioms to compare clusterings.
\cite{Ben2008Measures} introduced axioms on clustering quality measures.
\cite{zadeh2009uniqueness} presented a unique theorem for clustering
and showed that the impossibility result in \citep{kleinberg2003impossibility} can be avoided by relaxing one of clustering axioms.
\cite{ackerman2010towards} introduced several more properties of clustering functions to
taxonomize a set of popular clustering algorithmic paradigms
and strengthened the impossibility result on these properties.
\cite{gollapudi2009axiomatic} devised an evaluation method to characterize the axioms.

Following this line but in a more general sense,
\cite{teng2016itcs} recently initiated the axiomization of overlapping communities over preference networks.
A preference network is a finite set of individuals each having a preference indicating his/her affinity with others.
As an expressive model, it covers graphs, a structure that is commonly used in the field of community detection.
\cite{teng2016itcs} proposed a system of six axioms for community functions
and showed various properties of this system.

Among the properties, the lattice structure and the intersection lemma are most striking,
because they serve as a guideline to construct desirable community functions.
Specifically, the lattice structure means that the axioms satisfying all the six axioms form a lattice under the natural
``$\cup$" and ``$\cap$" operations.
Borgs et al.\ also identified the bottom (the minimum, called $\mathcal{C}_{cliq}$) and the top (the maximum, called $\mathcal{C}_{comp}$) of this lattice.
The intersection lemma claims that a community function satisfies all the axioms if and only if it is the intersection of $\mathcal{C}_{comp}$
and a community function satisfying \textit{four} axioms which are simple and easy to conform.
The significance of the intersection lemma is that it paves a way to construct a community function satisfying all the axioms,
while the lattice structure enables to derive more such community functions on this basis.

However, \cite{teng2016itcs} did not \textit{construct} any nontrivial community function satisfying all the axioms.
This is mainly due to the fact that it is coNP-complete to check two of the axioms, namely Group Stability ($\GS$) and Self-Approval ($\SA$).
Furthermore, $\GS$ and $\SA$ play critical roles in this axiomization for two reasons.
First, only they cover the stability of communities, while an unstable community does not make sense, especially in the context of preference networks.
As an analogue, consider the stable matching problem where affinity among the players is determined by their preferences and the stability of matchings is the focus (see e.g., \cite{gusfield1989stable}).
Second, these two axioms set the top of the lattice.
The coNP-completeness of the two most important axioms surely compromises the practicality of the axiom system.

On this ground, we make an effort to improve the axiom system in \citep{teng2016itcs} and construct desired community functions. Our contributions are fourfold.

\begin{enumerate}
 \item We show that axiom $\SA$ can degenerate so as to be checkable in polynomial time, leaving the axiom system equivalent.
 \item We propose an efficiently checkable alternative axiom $\SGS$ to the original $\GS$.
 The modification is natural, and preserves all the good properties of the original axiom system.
 \item We present a general approach to construct nontrivial community functions that conform to the new axiom systems.
 \item We find a natural, consistent, constructive, enumerable, and samplable community function,
 answering an open problem in \citep{teng2016itcs}.
 Note that both this community function and the just-mentioned general approach remain valid in the original axiom system,
 but the construction essentially depends on the new axiom $\SGS$ .
\end{enumerate}

The structure of this paper is as follows. In Section \ref{sec:pre},
we review some necessary notions and notations, and show that $\SA$ can be replaced by a simpler axiom.
We define a strengthening alternative to the crucial axiom $\GS$ in Section \ref{sec:poly}, named $\SGS$.
In Section \ref{Complexity_and_lattice},
we prove that $\SGS$ is efficiently checkable, and various desirable properties of the original axiom system are prereserved.
In Section \ref{Constructive_community_functions}, we present a systematic approach to finding consistent and constructive community functions,
and find an ideal community function.
Finally, we conclude this paper in Section \ref{sec:con} with future works and open problems.

\section{Preliminaries}\label{sec:pre}
In this paper, we consider community detection as a task that translates a preference network into a set of communities.
First of all, we review the definition of preference networks and the axiom system introduced by \cite{teng2016itcs}, and show some properties of the axiom system.

\subsection{Preference networks}
A preference network is a finite set of individuals each of which %is assigned a preference.
% The preference assigned to an individual
ranks all the individuals (including herself) in (strict) order of preference.
%:is a total order over all individuals indicating the affinity of the individual with others.

Formally, consider a non-empty set $V=\{1,2,\ldots, n\}$ and $L(V)$, the set of all total orders on $V$.

A total order $\pi\in L(V)$ is equivalently defined as a bijection $\pi: V\rightarrow V$, denoted by $[v_1 v_2 \ldots v_n]$
where $v_i=\pi^{-1}(i)$ for $i= 1,2,\ldots, n$.
For any $u,v\in V$ and $\pi\in L(V)$, $\pi(u)$ is referred to as the rank of $u$ in $\pi$,
and we say that $\pi$ prefers $u$ to $v$ if $\pi(u)<\pi(v)$, denoted by $u\succ_{\pi} v$.

A \emph{preference profile} $\Pi$ on $V$ is a list of total orders $\{\pi_u\}_{u\in V}\in L(V)^V$ on $V$. Given $u,v,w\in V$,
we say that $u$ prefers $v$ to $w$ if so does $\pi_u$, denoted by $v\succ_{\pi_u} w$.
Given a preference profile $\Pi$ on $V$, the
pair $N=(V, \Pi)$ is called a \emph{preference network}.

A \emph{community function} is a function $\mathcal{C}$ that maps each preference network $N=(V,\Pi)$ to a collection $\mathcal{C}(N)\subseteq 2^V$.
Each $S\in \mathcal{C}(N)$ is called a community of $N$ defined by $\mathcal{C}$, and the term ``defined by $\mathcal{C}$" will be omitted if clear from context.

As an example, consider the community function
% $$\mathcal{C}_{cliq}: (V,\Pi)\mapsto\{S\subseteq V: \forall u, v\in S, \forall w\in V-S, v\succ_{\pi_u}w\}.$$
$$\mathcal{C}_{cliq} (V,\Pi)=\{S\subseteq V: \forall u, v\in S, \forall w\in V-S, v\succ_{\pi_u}w\}.$$
Every community $S$ defined by $\mathcal{C}_{cliq}$ is called a clique\footnote{Note that a clique in a preference network does not mean it is a complete subgraph.}.
Roughly speaking, a clique is a subset of individuals preferring each other to those not in the subset.

\begin{example}\label{example:1}
We consider two families in a village and each family has three members.
The members in a family have the same preference and they prefer their own family members to other villagers.
In this case, every family is a community.
Formally,
suppose $N=(V,\Pi)$ is the preference network where $V=\{1, 2, 3, 4, 5, 6\}$, $\pi_1=\pi_2=\pi_3=[123456]$
and $\pi_4=\pi_5=\pi_6=[456123]$.
Sets $\{1,2,3\}$ and $\{4,5,6\}$ are in $\mathcal{C}_{cliq}(N)$.
\end{example}

Actually preference frameworks have been used in various occasions, e.g., voting \citep{gale1962college,gusfield1989stable,roth1984evolution},
network routing \citep{Rekhter1994A,caesar2005bgp},
and coalition formation in collaborative games \citep{brams2003dynamic,roth1984stable}.
\cite{balcan2013finding} further elaborated on how the concept of preference networks properly models social networks
and why it essentially captures their underlying community structures.

\subsection{Existing axioms}
\cite{teng2016itcs} proposed six axioms for community functions.
They are included here for the paper to be self-containing.
Let's arbitrarily fix a non-empty finite set $V$ and a community function $\mathcal{C}$.

\begin{axiom}[\textbf{Anonymity(A)}]
Let $S, S'\subseteq V$, and $\Pi, \Pi'$ be two preference profiles on $V$. Assume that $S'=\sigma(S)$ and $\Pi'=\sigma(\Pi)$ for some permutation $\sigma: V\rightarrow V$. Then $S\in \mathcal{C}(N)$ if and only if $S'\in \mathcal{C}(N')$,  where $N=(V, \Pi)$ and $N'=(V, \Pi')$.
\end{axiom}

This axiom means that whether $S$ is a community is determined by the preference profile rather than by labels of the individuals.

\begin{axiom}[\textbf{Monotonicity(Mon)}]
Let $S\subseteq V$, and $\Pi, \Pi'$ be two preference profiles on $V$. If $u\succ_{\pi_s}v\Rightarrow u\succ_{\pi'_s}v$ for $\forall s,u\in S$ and $\forall v\in V$, then $S\in \mathcal{C}(N)\Rightarrow S\in \mathcal{C}(N')$, where $N=(V, \Pi)$ and $N'=(V, \Pi')$.
\end{axiom}

Intuitively,
monotonicity means that if the change of the preference profile does not decrease the ranking of any member
of a community and preserves the relative ranking among the members, then the community remains. This is reasonable since the affinity among the members of the community is improved after the preference profile is changed.

\begin{axiom}[\textbf{Embedding(Emb)}]
 Given two preference networks $N=(V, \Pi)$ and $N'=(V', \Pi')$
 such that $V'\subseteq V$ and $\pi'_u(v)=\pi_u(v)$ for all $u, v\in V'$, then $\mathcal{C}(N')=\mathcal{C}(N)\cap 2^{V'}$.
\end{axiom}

Intuitively, though $N'$ is embedded into $N$, its preferences are not influenced.
Hence it is reasonable that communities inside
$V'$ are formed independently of individuals outside $V'$.

\begin{axiom}[\textbf{World Community(WC)}]
For any preference network $N=(V,\Pi)$, $V\in \mathcal{C}\left(N\right)$.
\end{axiom}

This axiom is self-explanatory.

Before introducing the remaining two axioms,
the concept of \emph{preference} should be extended to \textit{group-preference}\footnote{It is referred to as lexicographic preference in \citep{teng2016itcs}.} so that equal-sized sets can be compared.
Given a preference network $N=(V,\Pi)$ and
non-empty disjoint sets $G,G'\subseteq V$ of the same size,
we say that $s\in V$ \textit{group-prefers} $G'$ to $G$, denoted by $G' \succ_{\pi_s} G$, if one can align the elements
$g_1,g_2,\cdots,g_{|G|}$ of $G$ and $g'_1,g'_2,\cdots,g'_{|G|}$ of $G'$ so that
$g'_i\succ_{\pi_s}g_i$ for all $i=1,2,\cdots,|G|$.

For example, in the preference network in Example \ref{example:1},
% $\{1,4\}\succ_{\pi_2} \{2,5\}$ since $1\succ_{\pi_2}2$ and $4\succ_{\pi_2}5$.
% ***yuyi:0 ``we need to change the direction?''
% In fact, $\{1,4\}\succ_{\pi_i} \{2,5\}$ holds for all $1\le i\le 6$.
% ***yuyi:1
$\{3,4\}\succ_{\pi_2} \{5,6\}$ since $3\succ_{\pi_2}5$ and $4\succ_{\pi_2}6$.
But $\{3,4\}\succ_{\pi_5} \{5,6\}$ does not hold, because $5\succ_{\pi_5} 3$ and $6\succ_{\pi_5} 3$.

% Consider a preference network $N=(V,\Pi)$.
A set $S\subseteq V$ is called $self$-$approving$ with respect to $\Pi$ if for any $S'\subseteq V-S$ with $|S'|=|S|$,
$S'$ is not group-preferred to $S$ by some $s\in S$.
$S$ is called \emph{group-stable} with respect to $\Pi$
if for any non-empty $G\subsetneq S$ and $G'\subseteq V-S$ with $|G'|=|G|$,
$G'$ is not group-preferred to $G$ by some $s\in S-G$.
Roughly speaking, a set is group-stable if no subset agrees to leave and join another set.

%Consider a preference network $N=(V,\Pi)$.
%A set $S\subseteq V$ is called $group-stable$
%if for all non-empty $G\subsetneq S$ there exists no $G'\subseteq V-S$
%with the size of $|G|$ preferred to $G$ by all $s\in S-G$.
%$S$ is called $self$-$approving$  if there exists no $S'$ with the size of $S$ preferred to $S$ by all $s\in S$.

For instance, consider again the preference network in Example \ref{example:1}.
The set $S=\{2,5\}$ is not group-stable, since there exist $G=\{2\}$ and $G'=\{4\}$ such that $5\in S-G$ group-prefers $G'$ to $G$.
Actually, $S$ is not $self$-$approving$ either, because for any $s\in S$, $s$ group-prefers $\{1,4\}$ to $S$.

\begin{axiom}[\textbf{Self-Approval(SA)}]\label{axiom:SA}
    For any preference network $N=(V, \Pi)$, if $S\in \mathcal{C}(N)$, then $S$ is self-approving with respect to $\Pi$.
\end{axiom}

\begin{axiom}[\textbf{Group Stability(GS)}]
    Given a preference network $N=(V, \Pi)$, if $S\in \mathcal{C}(N)$, then $S$ is group-stable with respect to $\Pi$.
\end{axiom}

There exist community functions satisfying all the axioms,
e.g., $\mathcal{C}_{cliq}$ and $$\mathcal{C}_{comp}(N)=\{S\subseteq V: S\textrm{ is self-approving and group-stable}\}.$$

\subsection{Properties}

\cite{teng2016itcs} showed that this axiom system has nice properties.
The most impressive ones include the intersection lemma and the lattice structure.

\begin{lemma}[Lattice, \citep{teng2016itcs}]
Under the operations $\cap$ and $\cup$,
the community functions satisfying $\A$, $\Mon$, $\Emb$, $\WC$, $\GS$,
and $\SA$ form a lattice whose top is $\mathcal{C}_{comp}$ and bottom is $\mathcal{C}_{cliq}$.
\end{lemma}

Due to the lattice structure, given some community functions satisfying all the six axioms,
one can construct more such community functions using a mixture of $\cap$ and $\cup$ operations.

\begin{lemma}[Intersection Lemma, \citep{teng2016itcs}]
For any community function $\mathcal{C}$ satisfying $\A$, $\Mon$, $\Emb$ and $\WC$,
$\mathcal{C}\cap\mathcal{C}_{comp}$ satisfies all the six axioms.
\end{lemma}

By this lemma, designing a community function satisfying all the axioms is reduced to find one satisfying the four axioms, which is relatively
easier to be satisfied.
However, even we have %such
a community function satisfying the four axioms, it remains hard to check whether a given set is a community,
since it is computationally hard to check the key axioms $\GS$ and $\SA$.

\begin{theorem}[\citep{teng2016itcs}]\label{theorem_complex_gs}
It is coNP-complete to decide whether a subset $S\subseteq V$ is self-approval or group-stable with respect to a preference profile.
\end{theorem}

 Actually, we find that $\SA$ can degenerate so as to be efficiently checkable,
 keeping the axiom system equivalent to the original one.
 This is because $\GS$ almost implies $\SA$, which immediately follows from the next theorem.%as shown in the next theorem.

\begin{theorem}\label{GSimplySA}
Given a preference network $N=(V,\Pi)$ and a community $S\subseteq V$ with $|S|\geq 2$, if $S$ is group-stable with respect to $\Pi$, then it is self-approving with respect to $\Pi$.
%satisfies $\GS$, then $S$ satisfies $\SA$.
\end{theorem}
\begin{proof}
Suppose that $|S|\geq 2$ and $S$ is not self-approving with respect to $\Pi$.
Then there exists $S'\subseteq V\setminus S$ such that for all individuals $s\in S$, $s$ group-prefers $S'$ to $S$. Fix such a $S'$ and arbitrarily choose $s\in S$. Assume that $s'\in S'$ is the least favorite individual in $S'$ by the preference of $s$. We have that $s$ group-prefers $S'-\{s'\}$ to $S-\{s\}$.

%Then then $S$ is not group-stable with respect to $\Pi$.
%We prove this theorem by showing that if  $|S|\geq 2$ and $S$ is not self-approving with respect to $\Pi$, then $S$ is not group-stable with respect to $\Pi$.

%If $S$ does not satisfy $\SA$, then there exists $S'$ such that for all individuals $s\in S$, $s$ prefers $S'$ to $S$.
%Suppose $s'\in S'$ is the least favorite individual in $S'$ by $s$, then $s$ prefers $S'-\{s'\}$ to $S-\{s\}$.
Since $|S|\geq 2$, $S-\{s\}$ is a non-empty set. It follows that $S$ is not group-stable with respect to $\Pi$.
\end{proof}

Therefore, we present a degenerate version of $\SA$ so that it only deals with the case where $|S|=1$.
The degenerate version is

$\SAp$: For any preference network $N=(V, \Pi)$, if $\{s\}\in \mathcal{C}(N)$, then $\pi_s(s)=1$.

By Theorem \ref{GSimplySA}, any community function conforms with $\{\GS,\SA\}$ if and only if it conforms with $\{\GS,\SAp\}$.
Considering that $\SAp$ is efficiently checkable, we will use $\SAp$
% rather than
instead of $\SA$ in the remainder of the paper.

\section{Strengthening axiom $\GS$}\label{sec:poly}

After replacing $\SA$ with $\SAp$, axiom $\GS$ becomes the only difficult-to-check axiom in the axiom system.
It is this intractability that causes difficulty in constructing community functions.
In order to solve this problem, we modify $\GS$ in a natural way, resulting in an efficient-to-check axiom called $\SGS$.
Besides, we show that a community function satisfies $\GS$ if it satisfies $\SGS$ .

\subsection{Weak preferences}
The hardness of checking group stability of a subset $S$ is partially rooted at enumerating equal-sized subsets of $S$ and $V-S$.
A natural idea for tackling this problem is to compare subsets of $S$ directly with $V-S$, rather than with its subsets.
The first technical obstacle is to compare subsets of different sizes,
so we further extend the concept of preference.

\begin{definition}[Weak Preferences]
Suppose $S,S'\subseteq V$, $\pi\in L(V)$ and $k=\min(|S|,|S'|)$.
The preference profile $\pi$ \textit{\NewPrefers} $S$ to $S'$, denoted by $S\succ_{\pi}S'$,
if $\pi$ group-prefers the set of the top $k$ elements of $S$ to that of $S'$.
\end{definition}

Consider Example \ref{example:1} again. We can show that individual $4$ weakly prefers $\{4,1\}$ to $\{2,3,5\}$.
Specifically, since $\{2,3,5\}$ has three individuals and $\{4,1\}$ only has two, select the top two individuals $\{2,5\}$ from $\{2,3,5\}$ (according to the preference of $4$).
The result follows because $4$ group-prefers $\{4,1\}$ to $\{2,5\}$.

It is reasonable to borrow the notation ``$\succ$" from group-preference,
since weak preference is equivalent to group-preference if $|S|=|S'|$.
When $|S|>|S'|$, $S\succ_{\pi} S'$ means that there is $T\subsetneq S$ with $|T|=|S'|$ such that $\pi$ group-prefers $T$ to $S'$.
When $|S|<|S'|$, $S\succ_{\pi} S'$ means that $\pi$ group-prefers $S$ to any $T\subsetneq S'$ with $|S|=|T|$.

The following properties of weak preferences will be frequently used in the rest of this paper.
They immediately follow from the definition of weak preference, so the proofs are omitted.

\begin{lemma}\label{proposition_top_lexicographical_1}
    Given $\pi\in L(V)$, for any $S\subseteq V$ and $S'\subseteq U\subseteq V$ with $|S'|\geq |S|$, if $\pi$ \NewPrefers $S'$ to $S$, then $\pi$ \NewPrefers $U$ to $S$.
\end{lemma}

\begin{lemma}\label{proposition_top_lexicographical_2}
    Given $\pi\in L(V)$, for any $S, S'\subseteq V$, if $\pi$ \NewPrefers $S'$ to $S$, then for all $T\subseteq S$, $\pi$ \NewPrefers $S'$ to $T$.
\end{lemma}

In Example \ref{example:1}, suppose $S=\{2,4\}$ and $S'=\{1,3\}$,
then individual $4$ \NewPrefers $S'$ to $S$.
Suppose $U=\{1,3,5\}$ and $T=\{2\}$,
then $S'\subseteq U$ and $T\subseteq S$,
which means that individual $4$ \NewPrefers $U$ to $S$ by Lemma \ref{proposition_top_lexicographical_1}
and \NewPrefers $S'$ to $T$ by Lemma \ref{proposition_top_lexicographical_2}.

\subsection{Alternative to $\GS$}
Now we are ready to define the alternative $\SGS$ (namely strong group stability) to $\GS$.

Consider a preference network $N=(V,\Pi)$.
A set $S\subseteq V$ is called $strongly~group-stable$ with respect to $\Pi$
if for any non-empty $G\subsetneq S$,
there exists $s\in S-G$ which does not weakly prefers $V-S$ to $G$.
Strong group-stability differs from group-stability mainly in that weak preference rather than group-preference is used.

Let's arbitrarily fix a community function $\mathcal{C}$.
\begin{axiom}[$\textbf{SGS}$]
For any preference network $N=(V, \Pi)$ and any subset $S\subseteq V$, if $S\in \mathcal{C}(N)$, then $S$ is strongly group-stable with respect to $\Pi$.
\end{axiom}

Roughly speaking,
$\GS$ rules out a community $S$ if there exists a subset $T\subsetneq S$ such that all members of $S-T$ agree to replace $T$ with a common,
equal-sized group outside of $S$,
while $\SGS$ does so
under a weaker condition that essentially
captures two facts.
First, $T$ can be replaced by a group of a different size.
Second, members of $S-T$ might not have a common replacement of $T$ (in case that $|V-S|>|T|$),
but they agree on kicking $T$ out.
In this case, it is reasonable to rule out such communities; see the following example.

In Example \ref{example:1},
individual $4$ \NewPrefers $\{5,6\}$ to $\{1,2,3\}$,
which means that $\{1,2,3,4\}$ is not strongly group-stable and cannot be a community according to $\SGS$.
This coincides with commonsense,
since $4$ tends to leave $\{1,2,3\}$ and join $\{5,6\}$.
However, $\{1,2,3,4\}$ is group-stable and is a candidate of community in the original axiom system.

\begin{theorem}\label{proposition_gs_sgs_sa_ssa}
For any community function $\mathcal{C}$, if it satisfies $\SGS$, then it also satisfies $\GS$.
\end{theorem}
\begin{proof}
If some community function $\mathcal{C}$ satisfies $\SGS$ but does not satisfy $\GS$,
there must exist a preference network $N=(V,\Pi)$ and a community $S\in \mathcal{C}(N)$ which is strongly group-stable but not group-stable with respect to $\Pi$.
By the definition of group stability, there are two non-empty subsets $T\subsetneq S$ and $T'\subseteq V-S$
such that $|T'|=|T|$ and $u$ group-prefers $T'$ to $T$ for all individuals $u \in S-T$.
By Lemma \ref{proposition_top_lexicographical_1}, $u$ \NewPrefers $V-S$ to $T$, meaning that $S$ is not strongly group-stable. A contradiction is reached.
\end{proof}

\section{Complexity and lattice}\label{Complexity_and_lattice}
This section shows that the lattice structure, the intersection lemma, and richness of the axiom system in \cite{teng2016itcs} still hold in our axiom system.
More importantly, we show that it takes polynomial time to check whether a given community satisfies $\SGS$.

\subsection{Intersection lemma and lattice structure}

Let  $\mathcal{A}=\{\A, \Mon, \WC, \Emb, \GS, \SAp\}$ and $\mathcal{SA}=\{\A, \Mon, \WC, \Emb, \SGS, \SAp\}$.
For a collection $\mathcal{X} \in \{\mathcal{A},\mathcal{SA}\}$ of axioms, a community function is said to be $\mathcal{X}$-consistent if it conforms with all axioms in $\mathcal{X}$.

We go on with showing that many good properties of $\mathcal{A}$ are preserved by $\mathcal{SA}$.
An example is the intersection lemma, one of the most important properties of $\mathcal{A}$.
%Fortunately, it still holds in $\mathcal{SA}$.

Before introducing the intersection lemma,
we define the intersection between community functions.
Given two community functions $\mathcal{C}_1$ and $\mathcal{C}_2$,
we define $\mathcal{C}_1\cap\mathcal{C}_2$ as the community function that %for any preference network $N=(V, \Pi)$,
$\left(\mathcal{C}_1\cap\mathcal{C}_2\right)\left(N\right) = \mathcal{C}_1(N)\cap\mathcal{C}_2(N)$.
Similarly, one can define $\left(\mathcal{C}_1\cup\mathcal{C}_2\right)\left(N\right) = \mathcal{C}_1(N)\cup\mathcal{C}_2(N)$.

Then we define three community functions that will be used: $\mathcal{C}_{SGS},\mathcal{C}_{SA'}$, and $\mathcal{C}_{scomp}$,
which are the maximum community functions satisfying $\SGS$, $\SAp$, and both of them, respectively.
Formally, given a preference network  $N=(V,\Pi)$, $$\mathcal{C}_{SGS}(N)=\{S\subseteq V: S\textrm{ is strongly group-stable}\},$$
$$\mathcal{C}_{SA'}(N)=\{S\subseteq V: \textrm{ if }S=\{s\} \textrm{ then } \pi_s(s)=1\}, \textrm{ and }$$  $$\mathcal{C}_{scomp}=\mathcal{C}_{SGS}\cap\mathcal{C}_{SA'}.$$

\begin{lemma}[Intersection Lemma]\label{intersection_lemma}
If a community function $\mathcal{C}$ satisfying any axiom $X\in\{\A, \Mon, \WC, \Emb\}$, then $\mathcal{C}\cap\mathcal{C}_{scomp}$ satisfies $X$, $\SGS$ and $\SAp$.
\end{lemma}
\begin{proof}
Let $\tilde{\mathcal{C}}=\mathcal{C}\cap\mathcal{C}_{scomp}$.
If $X\in\{\A, \WC, \Emb\}$, $\tilde{\mathcal{C}}$ satisfies $X$ because $\mathcal{C}_{scomp}$ is obviously $\{\A, \WC, \Emb\}$-consistent. Hence, we just consider the case where $X=\Mon$.

Assume that preference networks $N=(V,\Pi)$, $N'=(V,\Pi')$, and community $S\in\tilde{\mathcal{C}}(N)$ satisfy the condition of $\Mon$, namely, for all $s,u\in S$ and $v\in V$, if $u\succ_{\pi_s}v$ then $u\succ_{\pi'_s}v$. %We proceed to show that $S\in\tilde{\mathcal{C}}(N')$.
The rest of the proof consists of two steps.

Step 1: Since $\mathcal{C}$ satisfies $\Mon$ and $S\in \mathcal{C}(N)$, we immediately have $S\in\mathcal{C}(N')$.

Step 2: We prove that $S\in\mathcal{C}_{scomp}(N')$ as following.

If $S$ is a singleton $\{s\}$, then $S\in \mathcal{C}_{scomp}(N)$ implies that $\pi_s(s)=1$, namely $s\succ_{\pi_s} u$ for any $u\neq s$. Hence, $s\succ_{\pi'_s} u$ for any $u\neq s$ and we have $S\in \mathcal{C}_{scomp}(N')$.

When $|S|>1$, $S\in \mathcal{C}_{scomp}(N')$ if and only if $S\in\mathcal{C}_{SGS}(N')$, so we only have to prove that $S\in\mathcal{C}_{SGS}(N')$. For contradiction, suppose $S\notin\mathcal{C}_{SGS}(N')$.
Then there exists a non-empty set $T\subsetneq S$ such that $V-S\succ_{\pi'_s}T$ for all $s\in S-T$. Thus, for all $s\in S-T$ and $1\leq i\leq k$ where $k=\min\{|T|, |V-S|\}$, the top-$i$th individual $u_i\in V-S$ and the top-$i$th individual $v_i\in T$ satisfy $u_i\succ_{\pi'_s} v_i$ . Recalling the assumption about $N,N',S$, we have $u_i\succ_{\pi_s} v_i$ for all $1\leq i\leq k$ and  $s\in S-T$, meaning that $V-S\succ_{\pi_s}T$ for all $s\in S-T$. This indicates that $S\notin\mathcal{C}_{SGS}(N)$. A contradiction is reached. Therefore, $S\in\mathcal{C}_{scomp}(N')$.

Altogether, we have shown that if $N,N'$ and $S\in\tilde{\mathcal{C}}(N)$ satisfy the condition of $\Mon$, then $S\in\tilde{\mathcal{C}}(N')$. This means that $\tilde{\mathcal{C}}$ satisfies $\Mon$, and the lemma follows immediately.
\end{proof}

In some sense,
the critical role of the intersection lemma lies in that it provides a normal form of consistent community functions:
the intersection of a simple community function with $\mathcal{C}_{scomp}$.
Hence finding a $\mathcal{SA}$-consistent community function is reduced to finding a $\{\A, \Mon, \WC, \Emb\}$-consistent one.
This reduction will greatly help us construct desirable community functions, as shown in Section \ref{subsection3.4}.
The intersection lemma also leads to the lattice structure of $\mathcal{SA}$-consistent community functions,
which is also true for $\mathcal{A}$ and is one of the most striking results in \citep{teng2016itcs}.

\begin{theorem}\label{theorem_intersection_lattice}
Let $\mathbb{C}$ be the class of all $\mathcal{SA}$-consistent community functions.
The algebraic structure $\mathcal{T} = \{\mathbb{C}, \cup, \cap\}$
forms a bounded lattice whose top and bottom are $\mathcal{C}_{scomp}$ and $\mathcal{C}_{cliq}$, respectively.
\end{theorem}

The proof of Theorem \ref{theorem_intersection_lattice} is straightforward by the intersection lemma (Lemma \ref{intersection_lemma})
and the proof of Theorem 3.1 in \cite{teng2016itcs}, so the detail is omitted.

The significance of the lattice structure lies in that given some $\mathcal{SA}$-consistent community functions,
one can produce more such community functions by a mixture of $\cap$ and $\cup$ operations.

According to the lattice structure, the richness of $\mathcal{SA}$-consistent community functions
is to some extent determined by the difference between $\mathcal{C}_{scomp}$ and $\mathcal{C}_{cliq}$.
Actually, the difference is significant, with evidence from the number of communities.
On the one hand, for any preference network $N$, the size of $\mathcal{C}_{cliq}(N)$ is linear,
by Lemma \ref{proposition_clique_num}. On the other hand, an example inspired by \citep{teng2016itcs} indicates that the size $\mathcal{C}_{scomp}(N)$
can be exponential for some preference network $N$.

\begin{lemma}\label{richness}
For any positive integer $n$, there exists a preference network $N=(V,\Pi)$ with $|V|=n$ such that $|\mathcal{C}_{scomp}(N)|\geq 2^{\lfloor n/2 \rfloor}$.
\end{lemma}
\begin{proof}
Let $m=\lfloor n/2 \rfloor$. Let $H=\{h_1, h_2, ...,h_m\}$ be a set of $m$ heroes, $S=\{s_1, s_2, ..., s_m\}$ be a set of $m$ sidekicks, and $V=H\cup S$.
For all $1\leq i\leq m$, define preferences $\pi_{h_i}=\pi_{s_i}=[h_i, h_1, h_2, ..., h_m, s_i, s_1, ...s_{i-1},s_{i+1}, ..., s_m]$. Let $\Pi$ be the preference profile consisting of all $\pi_{h_i}$'s and $\pi_{s_i}$'s. Consider the preference network $N=(V,\Pi)$. Arbitrarily choose a subset $U\subseteq S$.
We will prove that $H\cup U\in\mathcal{C}_{SGS}(N)$.

For contradiction, suppose $H\cup U\notin \mathcal{C}_{SGS}(N)$. Then there exists a non-empty set $T\subsetneq H\cup U$ such that $V-(H\cup U)\succ_{\pi_w}T$ for all $w\in (H\cup U)-T$.
Since any individual in $H$ is always preferred to any individual in $S$, it must hold that $T\subset U\subset S$, implying that $H\subseteq (H\cup U)-T$. Therefore, $V-(H\cup U)\succ_{\pi_h}T$ for any $h\in H$.
Now choose $1\leq i\leq m$ such that $s_i\in U$. We have $V-(H\cup U)\succ_{\pi_{h_i}}T$, contradictory to the fact that $s_i\succ_{\pi_{h_i}} v$ for any $v\in V-(H\cup U)\subseteq S$. Hence $H\cup U\in\mathcal{C}_{SGS}(N)$.

Since $|H\cup U|>1$,
by the definition of $\mathcal{C}_{scomp}$ and $\mathcal{C}_{\SAp}$, $H\cup U\in\mathcal{C}_{SGS}(N)$ implies that $H\cup U\in\mathcal{C}_{scomp}(N)$.
Because there are $2^m$ different $U$'s, $|\mathcal{C}_{scomp}(N)|\geq 2^m$.
\end{proof}

Since the size of $\mathcal{C}_{cliq}(N)$ is at most linear (proved in Lemma \ref{proposition_clique_num}),
Lemma \ref{richness} suggests that the collection of $\mathcal{SA}$-consistent community functions may be rich and the lattice may be far from collapsing.
The above lemma also implies that the though $\SGS$ is stronger than  $\GS$, the restriction is not too much.

\subsection{$\SGS$ is efficient to check}
Now we present the most important property of $\SGS$, that it can be efficiently checked. %Two lemmas will be used.
%Lemma \ref{proposition_top_lexicographical_1} shows that axiom $\SGS$ is more strict than its original version $\GS$;
%Lemma \ref{proposition_top_lexicographical_2} presents a property of weak preferences,
%which allows us designing a greedy algorithm that makes the new axiom efficiently checkable.
%The proofs of these two lemmas are straightforward, so they are omitted.
%We will show that it is %in polynomial time
%efficient to determine whether a given community $S$ satisfies $\SGS$.
%Remember that it is coNP-complete to $\GS$.
This is shown constructively, through an algorithm inspired by the following observation.

Consider a preference network $N=(V,\Pi)$ and a subset $S\subsetneq V$.
Suppose that $S\notin \mathcal{C}_{SGS}(N)$.
There must exist a non-empty set $T\subsetneq S$ that
all individuals in $S-T$ \NewPrefer $V-S$ to $T$.
If we extend any set $U\subseteq T$ by adding any $u\in S$ that does not \NewPrefer $V-S$ to $U$,
Lemma \ref{proposition_top_lexicographical_2} indicates that
no individual in $S-T$ can be added since all individuals in $S-T$ \NewPrefer $V-S$ to $U$.
So, if we happen to start with $U=\{u\}$ for some $u\in T$ and extend $U$ step-by-step,
the process must stop before $U=S$.

Based on this observation, we design an algorithm which for each individual $u\in S$,
initializes $U$ to be $\{u\}$, and iteratively extends $U$ by adding an individual in $S$ that does not \NewPrefer $V-S$ to $U$.
If and only if we get some $U\subsetneq S$ that cannot be extended any more, decide $S\notin \mathcal{C}_{SGS}(N)$.
The details are specified in Algorithm \ref{algorithm:dec_sgs_in}.

\begin{algorithm}
    \caption{Decide membership of $\mathcal{C}_{SGS}(N)$}
    \label{algorithm:dec_sgs_in}
    \textbf{Input} A preference network $(V,\Pi)$, a subset $S\subsetneq V$\\
    \textbf{Output} A boolean value indicating whether $S\in \mathcal{C}_{SGS}(N)$
    \begin{algorithmic}[1]
    \Procedure{Decide}{$V, \Pi, S$}%\Comment{Whether $S$ is in $C_{SGS}(V, \Pi)$}
        \For{$u\in S$}\label{deter_alg_tgs_traversal}
            \State $U\gets\{u\}$\label{init}
            \While{$U\not = S$} \label{algorithm:sgs_deter_while}
                \If{$\exists s\in S-U$ such that $V-S\nsucc_{\pi_s}U$}      \label{algorithm:if}
                       \State Arbitrarily choose such an $s$
                    \State $U \gets U\cup\{s\}$     \label{algorithm:deter_add_s}
                \Else
                    \State \Return $False$ \label{algorithm:returnfalse}
                \EndIf
            \EndWhile
        \EndFor
        \State \textbf{return} $True$ \label{algorithm:deter_return}
    \EndProcedure
    \end{algorithmic}
\end{algorithm}

We illustrate how Algorithm \ref{algorithm:dec_sgs_in} works using the preference network in Example \ref{example:1}.
Let $S$ be $\{1,2,3,4\}$.
At the beginning, $u$ is $1$ at Line \ref{deter_alg_tgs_traversal} and $U=\{1\}$.
Individual $2$ is added to $U$ at Line \ref{algorithm:deter_add_s} since
individual $2$ does not \NewPrefer $V-S=\{5\}$ to $U$.
Then $U=\{1,2\}$,
and individual $3$ is added to $U$ since individual $3$ does not \NewPrefer $V-S=\{5\}$ to $U$. Now $U=\{1,2,3\}$ and $S-U=\{4\}$.
Because individual 4 \NewPrefers $\{5\}$ to $U$,
the algorithm returns $False$ at Line \ref{algorithm:returnfalse}.

Now we prove the correctness of Algorithm \ref{algorithm:dec_sgs_in} and show that its running time is $O(|S|^4)$.

\begin{lemma}\label{lemma_deter_tgs}
    Given a preference network $N=(V, \Pi)$ and $S\subseteq V$, it is in time $O\left(|S|^4\right)$ to decide whether $S\in \mathcal{C}_{SGS}(N)$.
\end{lemma}
\begin{proof}
\textbf{Correctness: }

Suppose $S\notin\mathcal{C}_{SGS}(N)$. There must exist a non-empty set $T\subsetneq S$
that $V-S\succ_{\pi_s}T$ for all $s\in S-T$.
Note that the $For$ loop enumerates all $u\in S$.
If the algorithm terminates before any $u\in T$ is chosen at line \ref{init},
it surely outputs $False$.
Otherwise, we eventually have $u\in T$ at line \ref{init}.
By Lemma \ref{proposition_top_lexicographical_2}, $U\subseteq T$ always holds.
Hence, the $While$ loop will reach a state where $U\neq S$ and the condition at Line \ref{algorithm:if} is false.
Then, the algorithm also returns $False$.

On the other hand,
if the algorithm returns $False$,
all individuals in $S-U$ \NewPrefer $V-S$ to $U$,
which means that $S\notin\mathcal{C}_{SGS}(N)$.

\textbf{Complexity:}

There are two nested loops in Algorithm \ref{algorithm:dec_sgs_in}
and each loop runs at most $|S|$ times.
The complexity of checking the condition at line \ref{algorithm:if} is $O(|S|^2)$.
So, the complexity of Algorithm \ref{algorithm:dec_sgs_in} is $O(|S|^4)$.
\end{proof}

We are ready to show that $\mathcal{C}_{scomp}$ is efficiently checkable, in contrast to the coNP-completeness of checking the membership of $\mathcal{C}_{comp}$.

\begin{theorem}\label{theorem_time_complexity_deter_compre}
    Given a preference network $N=(V, \Pi)$ and $S\subseteq V$,
    it takes $O(|S|^4)$ time to decide whether $S\in \mathcal{C}_{scomp}(N)$ or not.
\end{theorem}
\begin{proof}
If $|S|\geq 2$, it is equivalent to decide whether $S\in \mathcal{C}_{SGS}(N)$, which can be done in $O(|S|^4)$ by Algorithm \ref{algorithm:dec_sgs_in}.

By Lemma \ref{lemma_deter_tgs},
it takes $O(|S|^4)$ time to decide whether $S\in \mathcal{C}_{SGS}$, which means that it is in $O(|S|^4)$ to decide whether $S\in \mathcal{C}_{scomp}$.

Otherwise, $S$ is a singleton $\{s\}$. In this case, $S\in \mathcal{C}_{scomp}(N)$ if and only if $\pi_s(s)=1$, which can be decided in constant time.
\end{proof}

\section{Constructing community functions}\label{Constructive_community_functions}
\cite{teng2016itcs} proposed an open problem: Is there an ideal community function in their axiom system $\mathcal{A}$?
In this section, we first solve the problem in the our axiom system $\mathcal{SA}$ by constructing such a community function.
Since $\mathcal{SA}$ is stronger than $\mathcal{A}$, this community function is also $\mathcal{A}$-consistent, answering the original open problem.

\subsection{Consistent, constructive community functions}
A community function $\mathcal{C}$ is said to be \textit{constructive}
if for any preference network $N=(V,\Pi)$,
the membership of $ \mathcal{C}(N)$ can be checked in polynomial-time in $|V|$.
Constructive community functions are desirable, but it is not easy to figure out a natural, nontrivial, $\mathcal{A}$-consistent one
(\textit{nontrivial} means different from $\mathcal{C}_{cliq}$).
\cite{teng2016itcs} mainly considered two families of candidates, namely $\mathcal{C}_{cliq(g)}$ and $\mathcal{C}_{harmonious(\lambda)}$ with non-negative function $g$
and real number $\lambda\in [0, 1]$.
They showed that though $\mathcal{C}_{cliq(g)}\cap \mathcal{C}_{comp}$ and $\mathcal{C}_{harmonious(\lambda)}\cap \mathcal{C}_{comp}$ are $\mathcal{A}$-consistent,
they are not constructive in general.
The definitions of $\mathcal{C}_{cliq(g)}$ and $\mathcal{C}_{harmonious(\lambda)}$ are presented here in order to make this paper self-contained.

\begin{definition}[$\mathcal{C}_{cliq(g)}$]
Given a non-negative function $g: \{1,2,3,...\}\rightarrow \{0,1,2,3,...\}$, for any preference network $N=(V,\Pi)$ and $S\subseteq V$,
$S\in\mathcal{C}_{cliq(g)}(N)$ if $\forall u,s\in S,\pi_s(u)\in [1: |S|+g(|S|)]$.
\end{definition}

\begin{definition}[$\mathcal{C}_{harmonious(\lambda)}$]\label{definition_harmonious_lambda}
Given $\lambda\in [0, 1]$, for any preference network $N=(V,\Pi)$ and $S\subseteq V$, $S\in\mathcal{C}_{harmonious(\lambda)}(N)$
if $\forall u\in S,v\in V-S$, at least a $\lambda$-fraction of $\{\pi_s:s\in S\}$ prefer $u$ to $v$.
\end{definition}

In the context of $\mathcal{SA}$, we have the following theorem which helps finding consistent constructive community function.
\begin{theorem}\label{theorem:construct_community_function}
If a constructive community function $\mathcal{C}$ is $\{\A$, $\Mon$, $\Emb$, $\WC\}$-consistent, then $\mathcal{C}\cap\mathcal{C}_{scomp}$ is $\mathcal{SA}$-consistent and constructive.
\end{theorem}
\begin{proof}
The theorem follows immediately from Lemma \ref{intersection_lemma} and Theorem \ref{theorem_time_complexity_deter_compre}.
\end{proof}

This theorem greatly simplifies the task of finding $\mathcal{SA}$-consistent constructive community functions,
since it is relatively easier to find a constructive one that is $\{\A$, $\Mon$, $\Emb$, $\WC\}$-consistent.

For example, for any non-negative function $g$ and real number $\lambda\in [0, 1]$,
 both $\mathcal{C}_{cliq(g)}$ and $\mathcal{C}_{harmonious(\lambda)}$ conform with $\A$, $\Mon$, $\Emb$, and $\WC$.
 Hence, Theorem \ref{theorem:construct_community_function} immediately implies the following lemma.
\begin{lemma}\label{proposition:sa_consistent_a_consistent}
The community functions $\mathcal{C}_{cliq(g)}\cap \mathcal{C}_{scomp}$ and $\mathcal{C}_{harmonious(\lambda)}\cap \mathcal{C}_{scomp}$ are
$\mathcal{SA}$-consistent and constructive, for any non-negative function $g$ and any real number $\lambda\in [0, 1]$.
\end{lemma}

Furthermore, by Theorem \ref{proposition_gs_sgs_sa_ssa},
the two families of community functions in Lemma \ref{proposition:sa_consistent_a_consistent} are also $\mathcal{A}$-consistent and constructive.

In the next subsection we go further in this direction and find a community function having more nice properties,
solving the open problem posed by \cite{teng2016itcs}.

\subsection{$\mathcal{C}_{grow}$: an ideal community function} \label{subsection3.4}
An ideal community function should not only be
constructive,
but also allows to efficiently enumerate the communities.
% As a result,
\cite{DBLP:journals/corr/BorgsCMT14} (the full version of \citep{teng2016itcs})
defined two more properties for community functions $\mathcal{C}$:

\begin{itemize}
\item \textbf{Samplable}: given a preference network $N=(V,\Pi)$,
% one can obtain a random sample of $\mathcal{C}(N)$ in time $P(|V|)$, where $P(\cdot)$ is a polynomial function;
one can randomly sample any community from $\mathcal{C}(N)$ in time $P(|V|)$,  where $P(\cdot)$ is a polynomial function;
% one can randomly sample a community from $\mathcal{C}(N)$ in time $P(|V|)$, where $P(\cdot)$ is a polynomial function, and every community can be selected with a strictly positive probability；
\item \textbf{Enumerable}: given a preference network $N=(V,\Pi)$,
one can enumerate $\mathcal{C}(N)$ in time $O(n^k |\mathcal{C}(N)|)$ for some constant $k$.
\end{itemize}

\cite{teng2016itcs} proposed the following open problem.
\begin{problem}[$\mathcal{A}\textbf{CCSE}$]\label{problem:itcs_ccse}
 Find a natural nontrivial community function that is $\mathcal{A}$-consistent, constructive, samplable, and enumerable.
 %Find a natural nontrivial community function that allows nontrivially overlapping communities and is $\mathcal{A}$-consistent, constructive, samplable, and enumerable. Here, two communities $A$ and $B$ are nontrivially overlapping if $A\cap B\neq \emptyset, A\nsubseteq B$, and $B\nsubseteq A$.
\end{problem}

%$\mathcal{C}_{cliq}$ is ruled out, because cliques are trivially overlapping as shown in the following lemma.

We solve the $\mathcal{A}\textbf{CCSE}$ problem in this subsection.
Actually, the following stronger problem is solved.

\begin{problem}[$\mathcal{SA}\textbf{CCSE}$]\label{problem:itcs_ccse}
Find a natural community function that allows nontrivially overlapping communities and is $\mathcal{SA}$-consistent, constructive, samplable, and enumerable.
\end{problem}

The basic idea of our solution is to reduce the problem according to Theorem \ref{theorem:construct_community_function}.
Namely, if we find a community function $\mathcal{C}$ satisfying
$\A$, $\Mon$, $\WC$, and $\Emb$, then $\mathcal{C}\cap\mathcal{C}_{scomp}$ is $\mathcal{SA}$-consistent.
Furthermore, if $\mathcal{C}$ is enumerable and $|\mathcal{C}(N)|$ is polynomial in $|V|$ for all preference networks $N=(V,\Pi)$,
$\mathcal{C}\cap\mathcal{C}_{scomp}$ will be enumerable and
samplable.

But how to find such a community function $\mathcal{C}$? Our approach is inspired by the method in \citep{palla2005uncovering}.
Roughly speaking, starting with best communities, i.e.,, cliques in a preference network,
we produce new good communities as many as possible by extending existing communities.
Intuitively, this growing process conforms with the formation of communities in real life (for example, consider how friendship forms).
The key of the approach is properly defining good communities.
For this end, we consider the following community function
which is adapted from $\mathcal{C}_{harmonious(\lambda)}$.

\begin{definition}[$\mathcal{C}_{harmon}$]\label{definition:harmon}
Given a preference network $N=(V,\Pi)$ and
a subset $S\subseteq V$, $S\in\mathcal{C}_{harmon}(N)$ if for all $u\in S$ and $v\in V-S$,
more than half of the members of $S-\{u\}$ prefer $u$ to $v$.
\end{definition}

The following lemma shows that $\mathcal{C}_{harmon}$ is efficiently checkable.

\begin{lemma}\label{lemma:time-of-harmon}
Given a preference network $N=(V,\Pi)$ and a set $S\subseteq V$,
it takes $O(|V|^3)$ time to determine whether $S$ is in $\mathcal{C}_{harmon}(N)$.
\end{lemma}
\begin{proof}
For each individual $u$ in $S$ and each individual $v$ in $V-S$, it takes $O(|S|)$ time to determine whether a majority of $S-\{u\}$ prefer $u$ to $v$,
which means that it takes $O(|V|^3)$ time to determine whether $S\in\mathcal{C}_{harmon}(N)$ or not.
\end{proof}

Based on the idea preceding Definition \ref{definition:harmon},
we define $\mathcal{C}_{grow}$ in terms of Algorithm \ref{algorithm_clique_spanning},
i.e., for any preference network $N$, $\mathcal{C}_{grow}(N):=\Call{CliqueGrowing}{N}$.

\begin{algorithm}
    \caption{The definition of $\mathcal{C}_{grow}$}
    \label{algorithm_clique_spanning}
    \textbf{Input} A preference network $N=(V,\Pi)$\\
    \textbf{Output} A collection $R$ of communities
    \begin{algorithmic}[1]
    \Procedure{CliqueGrowing}{$N$}%\Comment{Whether $S$ is in $C_{SGS}(V, \Pi)$}
        \State $C\gets {\mathcal{C}_{cliq}(N)}$    \label{algorithm_get_all_clique} \Comment{$C$ keeps the communities to be extended}
        \State $R\gets \emptyset$ \Comment{$R$ keeps the final communities}
        \While{$C$ is not empty}
            \State Arbitrarily choose an element $S$ from $C$  \label{algorithm_cliq_span_save_res}
            \State $C\gets C\setminus\{S\}$, $R\gets R\cup\{S\}$
            \For{$u$ in $V-S$}    \label{algirhtm_clique_spanning_traversal_s_1}
                \State $S'\gets S\cup\{u\}$
                \If{$S'\in \mathcal{C}_{harmon}(N)$} \label{algorithm_deter_harm}
                    \State $C\gets C\cup\{S'\}$
                \EndIf
            \EndFor
        \EndWhile
        \State \Return $R$
    \EndProcedure
    \end{algorithmic}
\end{algorithm}

Algorithm \ref{algorithm_clique_spanning} is self-explanatory. We now show that the size of $\mathcal{C}_{grow}(N)$ and the time complexity of Algorithm \ref{algorithm_clique_spanning} are both polynomial in $|V|$.
%The initialization stage (Line \ref{algorithm_get_all_clique}) of Algorithm \ref{algorithm_clique_spanning}
%lets $C$ be $\mathcal{C}_{cliq}$ and the result set $R$ be an empty set.
%At Line \ref{algorithm_cliq_span_save_res},
%this algorithm arbitrarily chooses an element $S$ in $C$, and then removes it from $C$,
%adds it into the set $R$,
%and tries to extend it to get a new set
%at Line \ref{algirhtm_clique_spanning_traversal_s_1}.
%If we find an individual $u$ that $S\cup\{u\}$ is in $\mathcal{C}_{harmon}(V,\Pi)$, then we put it into $C$.
\begin{lemma}\label{SizeAndComplexity}
For any preference network $N=(V,\Pi)$, the size of $\mathcal{C}_{grow}(N)$ is $O(|V|^2)$, and Algorithm \ref{algorithm_clique_spanning} terminates within time $O(|V|^6)$.
\end{lemma}
\begin{proof}
According to the definition of $\mathcal{C}_{harmon}$, given $S\subseteq V$, if $S\cup \{u\}\in \mathcal{C}_{harmon}(N)$ for some $u\in V-S$, then for any $v\in V-S$ with $v\neq u$, a majority of individuals in $S$ prefer $u$ to $v$. As a result, for any $v\in V-S$ with $v\neq u$, $S\cup \{u\}\notin \mathcal{C}_{harmon}(N)$. This means that for any $S\in C$, it can be extended by at most one $u\in V-S$ in the \textbf{for} loop. Hence, starting with any $S\in\mathcal{C}_{cliq}$,
at most $O(|V|)$ communities can be obtained. Since there are only $O(|V|)$ cliques by Lemma \ref{proposition_clique_num}, the size of $\mathcal{C}_{grow}(N)$ is $O(|V|^2)$.

Then we discuss the time complexity of Algorithm \ref{algorithm_clique_spanning}. The \textbf{If} statement at line \ref{algorithm_deter_harm}
is executed for $|V-S|$ times for each community $S$. This fact, together with Lemma \ref{lemma:time-of-harmon},
implies that the time complexity of Algorithm \ref{algorithm_clique_spanning} is $O(|V|^6)$.
\end{proof}

Lemma \ref{SizeAndComplexity} immediately indicates that $\mathcal{C}_{grow}$ is enumerable, samplable, and constructive. However, $\mathcal{C}_{grow}$ is not $\mathcal{SA}$-consistent. For example, in the preference network of Example \ref{example:1}, the set $\{1,2,3,4\}$ is a community defined by $\mathcal{C}_{grow}$, but it is not strongly group-stable since individual $4$ \NewPrefers $\{5,6\}$ to $\{1,2,3\}$.

Fortunately, $\mathcal{C}_{grow}\cap \mathcal{C}_{scomp}$ is a community function satisfying all the requirements in Problem \ref{problem:itcs_ccse}.

\begin{theorem}\label{theorem_c_span_samplale_enumrable}
$\mathcal{C}_{grow}\cap\mathcal{C}_{scomp}$ is a natural community function that allows nontrivially overlapping communities and solves both the $\mathcal{SA}$\textbf{CCSE} problem and the $\mathcal{A}\textbf{CCSE}$ problem.
\end{theorem}
\begin{proof}
By Lemma \ref{SizeAndComplexity}, we know that $\mathcal{C}_{grow}$ is enumerable, samplable, and constructive, and that the size of $\mathcal{C}_{grow}(N)$ is polynomial for any preference network $N$. Since $\mathcal{C}_{scomp}$ is constructive, $\mathcal{C}_{grow}\cap\mathcal{C}_{scomp}$ is enumerable, samplable, and constructive.

It is straightforward to show that $\mathcal{C}_{grow}$ satisfies $\A$, $\Mon$, $\WC$, and $\Emb$. According to Lemma \ref{intersection_lemma}, $\mathcal{C}_{grow}\cap\mathcal{C}_{scomp}$ is $\mathcal{SA}$-consistent.

By Lemmas \ref{proposition_clique_num} and \ref{proposition:lower_bound_on_grow}, $\mathcal{C}_{grow}\cap\mathcal{C}_{scomp}\neq\mathcal{C}_{cliq}$, namely, $\mathcal{C}_{grow}\cap\mathcal{C}_{scomp}$ is nontrivial.

As a result,  $\mathcal{C}_{grow}\cap\mathcal{C}_{scomp}$ solves the $\mathcal{SA}$\textbf{CCSE} problem. By Theorem \ref{proposition_gs_sgs_sa_ssa}, it also solves the $\mathcal{A}\textbf{CCSE}$ problem.
\end{proof}

Now we prove Lemmas \ref{proposition_clique_num} and \ref{proposition:lower_bound_on_grow} just used. First, recall a property of $\mathcal{C}_{cliq}$.

\begin{lemma}[\citep{teng2016itcs}]\label{lemma_clique_compatible}
Given a preference network $N=(V, \Pi)$ and $S_1, S_2\in \mathcal{C}_{cliq}(N)$,
then either $S_1\subseteq S_2$, $S_2\subseteq S_1$, or $S_1\cap S_2 = \emptyset$.
\end{lemma}

\begin{lemma}\label{proposition_clique_num}
For any preference network $N=(V, \Pi)$, the size of $\mathcal{C}_{cliq}(N)$ is $O(|V|)$.
\end{lemma}

\begin{proof}%
%Suppose $\Pi$ is the preference profile that the size of $\mathcal{C}_{cliq}(V, \Pi)$ is maximum and
%$S$ is a community in $\mathcal{C}_{cliq}(V,\Pi)$ such that there exists no community $S'$ in $\mathcal{C}_{cliq}(V, \Pi)-\{V\}$ that $S\subseteq S'$.
%%Let $\Pi^{S}$ be the preference profile
%Because $V\in \mathcal{C}_{cliq}(V,\Pi)$ and $\forall S_1, S_2\in \mathcal{C}_{cliq}(V,\Pi)$,
%$S_1\cap S_2\neq \emptyset$ means $S_1\subseteq S_2~or~S_2\subseteq S_1$ by Lemma \ref{lemma_clique_compatible},
%$|\mathcal{C}_{cliq}(V,\Pi)|\leq |\mathcal{C}_{cliq}(S,\Pi^{S})| + |\mathcal{C}_{cliq}(V-S,\Pi^{V-S})|+1$.
%
%Since $|\mathcal{C}_{cliq}(U,\Pi^{U})|\leq 1$ for any $U\subseteq V$ with $|U|=1$,
%we assume $|\mathcal{C}_{cliq}(U,\Pi)|\leq 2|U|-1$ for any $U\subset V$.
%Then
%\begin{align*}
%|\mathcal{C}_{cliq}(V,\Pi)|&\leq |\mathcal{C}_{cliq}(S,\Pi^{S})| + |\mathcal{C}_{cliq}(V-S,\Pi^{V-S})|+1\\&\leq 2|S|-1+2|V-S|-1+1=2|V|-1,
%\end{align*} which means that the size of $\mathcal{C}_{cliq}(V,\Pi)$ is linear to the size of $V$.
For any preference network $N=(V, \Pi)$, construct a graph $T(N)$ corresponding to
$\mathcal{C}_{cliq}(N)$ as follows. Each vertex of $T(N)$ stands for a clique in $\mathcal{C}_{cliq}(N)$,
and an edge between two vertices exists if and only if the corresponding two cliques satisfy the condition
that one is a maximal clique inside the other. By Lemma \ref{lemma_clique_compatible},
$T(N)$ is a tree. Now construct another tree $T'(N)$ by extending $T(N)$ in this way: for any vertex $u$ of $T(N)$,
if the corresponding clique $S$ is such that there is only one maximal clique $S'$ inside $S$,
add a virtual vertex $v$ corresponding to $S-S'$ and add an edge between $u,v$. Obviously, $T'(N)$ remains a tree and the number of vertices of $T'(N)$ is no smaller than that of $T(N)$.

Now we view $T'(N)$ as a rooted tree whose root is the vertex corresponding to $V$.

The rooted tree  $T'(N)$ has two properties. First, every inner vertex has degree at least 2.
Second, the leaves are disjoint subsets of $V$, implying that there are at most $|V|$ leaves.
It is easy to see that such a tree has at most $|V|-1$ inner vertices.
As a result, the number of vertices of $T'(N)$ is at most $2|V|-1$, meaning that the size of $\mathcal{C}_{cliq}(N)$ is at most $2|V|-1$.
\end{proof}

\begin{lemma}\label{proposition:lower_bound_on_grow}
For any finite set $V$,
there exists a preference profile $\Pi$
such that the size of $\mathcal{C}_{grow}\cap\mathcal{C}_{scomp}(N)$
is $\Omega(|V|\log|V|)$, where $N=(V,\Pi)$.
\end{lemma}

This lemma will be proved constructively. Before presenting the proof, we show the basic idea of the construction, and briefly explain how the construction algorithm works.
%Given an individual set $V$,
%suppose $\Pi$ is the output of \Call{$GetProfile$}{$V$}
%and $N=(V,\Pi)$.
%We show that the size of $\mathcal{C}_{grow}\cap\mathcal{C}_{scomp}(N)$
%is $\Omega(|V|\log|V|)$.

Intuitively,
to maximize the size of $\mathcal{C}_{grow}\cap\mathcal{C}_{scomp}(N)$,
we should maximize the size of $\mathcal{C}_{cliq}(N)$.
Besides, in order to absorb as many individuals as possible,
for any two members inside a community, they should have the same preference order on the individuals outside of the community.
For example, let $N=(V,\Pi)$ be a preference network where $V=\{1,2,3,4,5\}$ and $\pi_1=[1,2,3,4,5]$, $\pi_2=[2,1,3,4,5]$,
$\pi_3=[3,4,1,2,5]$, $\pi_4=[4,3,1,2,5]$ and $\pi_5=[1,2,3,4,5]$,
then $\{1,2,3\}$ and $\{3,4,1\}$ are in $\mathcal{C}_{grow}\cap\mathcal{C}_{scomp}(N)$ but not in $\mathcal{C}_{cliq}(N)$.

Following the above idea, we design Algorithm \ref{set_preference_to_maximum_grow} which outputs the desired preference profile.
Let's first introduce the notations.

Throughout the algorithm, $\sigma$ stands for an arbitrarily fixed total order on $V$.
For any $S\subseteq V$, $\sigma |_S$ is defined to be the total order %on $S$ obtained by restricting $\sigma$
restricted to $S$.
For example, if $S = \{1, 2, 3, 5\}\subset V=\{1,2,...5\}$ and $\sigma=[1,4,3,2,5]$, then $\sigma|_{S}$ is $[1,3,2,5]$.
For a preference network $N=(V,\Pi)$ and $v\in V$, we use $\Pi|_v$ to stand for the preference of $v$ in $N$.
For any disjoint finite sets $W',W''$ and preferences $\pi'\in L(W'),\pi''\in L(W'')$,
the \textit{concatenation} of $\pi'$ with $\pi''$, denoted by $\pi'\bullet\pi''$, is defined to be the preference $\pi\in L(W'\cup W'')$ as follows:
$\forall w'\in W'$, $\pi(w')=\pi'(w')$ and $\forall w''\in W''$, $\pi(w'')=|W'|+\pi''(w'')$.
Intuitively, $\pi$ locally preserves the ordering of $\pi'$ and $\pi''$ on $W'$ and $W''$, but globally prefers $W'$ to $W''$.
\begin{algorithm}
    \caption{Find a preference profile with big $\mathcal{C}_{grow}\cap\mathcal{C}_{scomp}$}
%        \caption{Find a preference profile to get lower bound of $|\mathcal{C}_{grow}\cap\mathcal{C}_{scomp}|$}
    \label{set_preference_to_maximum_grow}
    \textbf{Input} A set $V$\\
    \textbf{Output} A preference profile $\Pi$ on $V$
    \begin{algorithmic}[1]
    \Procedure{$GetProfile$}{$V$}
        \If{$V$ is a singleton $\{v\}$}  $\pi_v\gets [v]$
        \Else
        \State Partition $V$ into $V_1$ \& $V_2$ s.t. $0\leq |V_1|-|V_2|\leq 1$
        %equally divide $V$ into two subsets $V_1$, $V_2$  \label{algorithm_set_pref_split} %such that $||V_1|-|V_2||\leq 1$
        %Split V into two subset~V_1~V_2~such~that~|V_1|-|V_2|\in [-1, 1]$
        \State $\Pi_1\gets \Call{GetProfile}{V_1}$  \label{algorithm_set_pref_get_subset_pref_1}
        \State $\Pi_2\gets \Call{GetProfile}{V_2}$  \label{algorithm_set_pref_get_subset_pref_2}
        \For{$v\in V_1$}
            \State $\pi_v\gets \Pi_1|_v\bullet (\sigma|_{V_2})$% to be such that $\forall v\in V_1$, $\pi_s(v)=\Pi_1|_s(v)$ and $\forall v\in V_2$, $\pi_s(v)=\sigma|_{V_2}(v)+|V_1|$           \label{algorithm:set_pref_in_u1}
        \EndFor
        \For{$v\in V_2$}
            \State $\pi_v\gets \Pi_2|_v\bullet (\sigma|_{V_1})$% to be such that $\forall v\in V_1$, $\pi_s(v)=\Pi_1|_s(v)$ and $\forall v\in V_2$, $\pi_s(v)=\sigma|_{V_2}(v)+|V_1|$           \label{algorithm:set_pref_in_u1}
        \EndFor
%        \State Likewise set $\pi_t$ for all $t\in V_2$% as $\pi_s$ for $s\in U_1$.
        \EndIf
        \State \Return $\Pi=\{\pi_v\}_{v\in V}$
    \EndProcedure
    \end{algorithmic}
\end{algorithm}

The procedure \Call{$GetProfile$}{$V$} is in a divide-and-conquer style. It divides $V$ into two subsets $V_1$ and $V_2$ which are balanced in size. Then it
recursively determine the preference profiles $\Pi_1$ and $\Pi_2$ on $V_1$ and $V_2$, respectively. %at Line \ref{algorithm_set_pref_get_subset_pref_1} and \ref{algorithm_set_pref_get_subset_pref_2}. At Line \ref{algorithm:set_pref_in_u1}, for each individual $s$ in $U_1$,
Finally,  the preferences in $\Pi_1$ and $\Pi_2$ are extended, resulting in the preference profile $\Pi$ on $V$. The extension of each preference in $\Pi_i$ is by concatenating it with $\sigma|_{V_{3-i}}$, for $i\in\{1,2\}$.
%Basically, each preference in $\Pi_1$ is extended by concatenating it with $\sigma|_{V_2}$, while each preference in $\Pi_2$ is extended by appending $\sigma|_{V_1}$ at the end. F%is such that   preserves its ordering on $U_1$ but appends individuals in $U_2$ we set $\pi_s$ to make sure $\forall u\in U_1$, $\pi_s(u)=\Pi_1|_s(u)$ and $\forall u\in U_2$, $\pi_s(u)=\sigma|_{U_2}(u)+|U_1|$.We set the preferences of individuals in $U_2$ in a similar way.

Now we prove that the preference profile $\Pi$ meets the requirement of Lemma \ref{proposition:lower_bound_on_grow}.

\begin{proof}
Fix $V$ throughout this proof. %Without loss of generality, assume that $|V|=$
Let  $N=(V,\Pi)$ with $\Pi$ being the preference profile found by the algorithm. For any set $U\subseteq V$ that appears in the recursion, if $|U|>1$, the algorithm divides it into $U_1$ and $U_2$. Let $\mathcal{K}(U)=(\mathcal{C}_{grow}\cap\mathcal{C}_{scomp})(N)\cap 2^U$, the set of communities of $N$ inside $U$.
Define $F(U)=|\mathcal{K}(U)|$, and
$g(U)=|\mathcal{K}(U)\setminus (2^{U_1}\cup 2^{U_2})|$. Obviously, $F(U)=F(U_1)+F(U_2)+g(U)$ always holds. The show that $g(U)$ is linear in the size of $U$, we make two claims.

\textbf{Claim 1}: $|(\mathcal{C}_{grow}(N)\cap 2^U)\setminus (2^{U_1}\cup 2^{U_2})|=\Theta(|U|)$.

This claim can be proved in two steps.

First, by Algorithm \ref{set_preference_to_maximum_grow}, each clique $C\in \mathcal{C}_{cliq}(N)\cap 2^U$ with $C\neq U$ must satisfy either $C\subseteq U_1$ or $C\subseteq U_2$. Recalling the \textbf{for} loop of Algorithm \ref{algorithm_clique_spanning}, we have that if $S\in \mathcal{C}_{grow}(N)$  is obtained by extending $S'\subsetneq U_1$ with some $u\in V-S'$, then $u\in U_1$. So,  any community $S\in(\mathcal{C}_{grow}(N)\cap 2^U)\setminus (2^{U_1}\cup 2^{U_2})$ must be extended from (hence include) either $U_1$ or $U_2$.

Second, recall how Algorithm \ref{algorithm_clique_spanning} extends $U_1$ with individual in $U_2$. Suppose $\sigma|_{U_2}=[u_1,...u_{|U_2|}]$. Since $|U_1|\geq |U_2|$ and $\pi_v|_{U_2}=\pi_{v'}|_{U_2}=\sigma|_{U_2}$ for any $v,v'\in U_1$, one knows that for any community $S\in2^U$ with $U_1\subset S$, it is a community in $\mathcal{C}_{grow}(N)$ if and only if $S=U_1\cup\{u_1,...u_k\}$ for some $k\leq |U_2|$. Likewise, any community $S\in2^U$ with $U_2\subset S$ is a community in $\mathcal{C}_{grow}(N)$ if and only if $S=U_2\cup\{v_1,...v_k\}$ for some $k\leq |U_1|$, where the $v_i\in U_1$ are such that $\sigma|_{U_1}=[v_1,...v_{|U_1|}]$.

As a result, Claim 1 holds.

\textbf{Claim 2}: If $U\neq V$, there are at most one community in $(\mathcal{C}_{grow}(N)\cap 2^U)\setminus (2^{U_1}\cup 2^{U_2})\setminus \mathcal{C}_{scomp}(N)$.

This claim can be proved as follows. Arbitrarily choose $S\in(\mathcal{C}_{grow}(N)\cap 2^U)\setminus (2^{U_1}\cup 2^{U_2})$. We begin with the case $U_1\subseteq S$. Suppose that $S\notin\mathcal{C}_{scomp}(N)$. Then there is $T\subsetneq S$ such that $V-S\succ_{\pi_u} T$ for any $u\in S-T$. Since $U\neq V$, $|V-U|\geq |U|-1\geq |T|$. We again go in two steps.

First, arbitrarily choose $v\in T\cap U_1$ if $T\cap U_1\neq \emptyset$, otherwise choose $v\in T\cap U_2$. Suppose that $U_1\nsubseteq T$; arbitrarily choose $u\in U_1\setminus T$. By Algorithm \ref{set_preference_to_maximum_grow} and the characterization of $(\mathcal{C}_{grow}(N)\cap 2^U)\setminus (2^{U_1}\cup 2^{U_2})$ in the proof of Claim 1, we have $v\succ_{\pi_u} w$ for any $w\in V-S$. Considering that $|V-S|\geq |V-U|\geq |U|-1\geq |T|$ and $V-S\succ_{\pi_u} T$, it holds that $w\succ_{\pi_u} v$ for some $w\in V-S$, which is a contradiction. Hence, $U_1\subseteq T$.

Second, arbitrarily choose $u\in S-T\subseteq U$. %we now have $S-T\subseteq U_2$. Because $U\neq V$, $V-U\neq\emptyset$.
By Algorithm \ref{set_preference_to_maximum_grow}, for any $v\in U,w\in V-U$, it holds that $v\succ_{\pi_u} w$. %In additional, $V-S=(U_2-S)\cup V-U,T=(T-U_1)\cup U_1$. Since $U_1\subseteq T\subsetneq S\subseteq U_1\cup U_2$, one has $|U_2-S|<|U_2|\leq |U_1|$.
Recall that $V-S\succ_{\pi_u} T$ and $|V-S|\geq |T|$, so for any $v\in T$, there is a $w_v\in V-S$ such that $w_v\succ_{\pi_u} v$. Because $V-S=(U_2-S)\cup (V-U),T=(T-U_1)\cup U_1$, and $|U_2-S|<|U_2|\leq |U_1|$, we have $w_v\in V-U$ for some $v\in T$, contradictory to the fact that $v\succ_{\pi_u} w$ for any  $v\in T\subseteq U,w\in V-U$.

Therefore, for any $S\in(\mathcal{C}_{grow}(N)\cap 2^U)\setminus (2^{U_1}\cup 2^{U_2})$, if $U_1\subseteq S$, we get $S\in\mathcal{C}_{scomp}(N)$.

Likewise, if $U_2\subseteq S$, we can show that $S\in\mathcal{C}_{scomp}(N)$ except the unique case $S=\{u\}\cup U_2$ with $u\in U_1$.

As a result, by Claim 1, Claim 2 holds.

Claims 1 and 2 implies that $g(V)=O(|V|)$ and for any set $U\subsetneq V$ that appears in the recursion of Algorithm \ref{set_preference_to_maximum_grow},   $g(U)=\Theta(|U|)$. A simple calculation indicates that $F(V)=\Theta(|V|\log |V|)$.
\end{proof}

%Since the number of cliques in any preference network is $O(|V|)$,
%the community function $\mathcal{C}_{grow}\cap\mathcal{C}_{scomp}$ far differs from $\mathcal{C}_{cliq}$.
%Altogether, Theorem \ref{theorem_c_span_samplale_enumrable} is proved.
%
%So, the problem $\mathcal{SA}\textbf{CCSE}$ is solved. By Theorem \ref{proposition_gs_sgs_sa_ssa}, $\mathcal{C}_{grow}\cap\mathcal{C}_{scomp}$ also solves the problem $\mathcal{A}\textbf{CCSE}$.

\begin{remark}
Actually, \cite{teng2016itcs} did not define enumerable or samplable community functions, but we find the definitions in \citep{DBLP:journals/corr/BorgsCMT14},
the arXiv version of \citep{teng2016itcs}. \cite{teng2016itcs} also defined the stability and
the open problem required that the community function has stable communities which are samplable and enumerable.
We find that even $\mathcal{C}_{cliq}$ is not stable under their definition, so we just ignore it.
\end{remark}

\begin{remark}
Lemma \ref{intersection_lemma} (the intersection lemma) plays an important role in solving the problem $\mathcal{SA}\textbf{CCSE}$
because $\mathcal{C}_{grow}$ itself does not satisfy $\SGS$.
To see this, consider Example \ref{example:1}.
Since $\{1, 2, 3\}$ is in $\mathcal{C}_{cliq}(N)$ and $\{1,2,3, 4\}$ is in $\mathcal{C}_{harmon}(N)$,
$\{1,2,3,4\}$ is in $\mathcal{C}_{grow}(N)$,
while it is not in $\mathcal{C}_{SGS}(V,\Pi)$ because $\pi_4$ weekly prefers $\{5\}$ to $\{1,2,3\}$.
\end{remark}

According to Theorem \ref{theorem_complex_gs},
if one uses axioms \emph{GS} and \emph{SA},
then it is coNP-complete to determine whether a given community is in $\mathcal{C}_{grow}\cap\mathcal{C}_{GS}\cap\mathcal{C}_{SA}(N)$.
This example also shows the importance of getting rid of the computational difficulties of checking the key axioms of \emph{GS} and \emph{SA}.

\section{Conclusion}\label{sec:con}

In this paper, we focus on axiomization of network community detection. As far as we know, only one paper \citep{teng2016itcs} in this line studied overlapping communities, and it was on a general structure -- preference networks. Among the six axioms in \citep{teng2016itcs}, two of them ($\GS$ and $\SA$) play a critical role, but the hardness in checking them compromises the practicality of the axiom system.
We showed that in the context of the axiom system, $\SA$ can be equivalently replaced by a degenerate version $\SAp$.
We also naturally modified axiom $\GS$ to a stronger version called $\SGS$.
We showed that both $\SGS$ and $\SAp$ can be checked in polynomial time, and most of the properties of the original system are preserved.
Furthermore, by the intersection lemma, we found two constructive and $\mathcal{SA}$-consistent community functions.
We also found an $\mathcal{SA}$-consistent, constructive, samplable and enumerable community function
that allows nontrivially overlapping communities,
thus answering to an open problem in \citep{teng2016itcs}.

Although the current work is purely theoretical,
it would be beneficial to evaluate and improve existing community detection algorithms in the framework of our axioms.
This is a direction of our future work.

Another direction is to further improve the axiom system. Before Lemma \ref{richness},
we discussed the richness, namely, many community functions might be consistent with the axioms.
This seems advantageous since good community functions are not likely to be precluded.
However, selectivity should also be considered,
otherwise users would be burdened with selecting desirable community functions from too many candidates.
Can we further improve the system of axiom by making a better trade-off between richness and selectivity?
\section*{Acknowledgment}
The authors would like to thank Shanghua Teng at USC for introducing the topic and having helpful online meetings with us. We also thank Wei Chen at MSR Asia and Pinyan Lu at Shanghai University of Finance and Economics for their valuable advice at the beginning of our work.

\bibliography{reference}
\bibliographystyle{named}
\end{document}